\documentclass[12pt]{amsart}

\usepackage{amsmath,amssymb,url}
\usepackage{mathrsfs,color}
\usepackage{enumerate,graphicx}
\usepackage{epstopdf}
\usepackage{color,soul}
\usepackage{tcolorbox}
\usepackage{a4wide}
\usepackage{hyperref}

\newtheorem{theorem}{Theorem}[section]

\newtheorem{lemma}[theorem]{Lemma}

\newtheorem{definition}[theorem]{Definition}

\newtheorem{remark}{Remark}

\newcommand{\alphamin}{{\alpha_{\min}}}
\newcommand{\alphamax}{{\alpha_{\max}}}

\title[Fractional modelling of COVID-19 transmission]{%
Fractional modelling of COVID-19 transmission incorporating asymptomatic and super-spreader individuals}

\author[M.~Khalighi]{Moein Khalighi}
\address[M.~Khalighi]{Department of Computing, 
University of Turku, Turku, Finland}
\email{moein.khalighi@utu.fi}

\author[L.~Lahti]{Leo Lahti}
\address[L.~Lahti]{Department of Computing, 
University of Turku, Turku, Finland}
\email{leo.lahti@utu.fi}

\author[F.~Nda\"{\i}rou]{Fa\"{\i}\c{c}al Nda\"{\i}rou}
\address[F.~Nda\"{\i}rou]{Institute of Mathematics and Informatics, 
Bulgarian Academy of Sciences, Sofia 1113, Bulgaria}
\email{faical@math.bas.bg}

\author[P.~Rashkov]{Peter Rashkov}
\address[P.~Rashkov]{Institute of Mathematics and Informatics, 
Bulgarian Academy of Sciences, Sofia 1113, Bulgaria}
\email{p.rashkov@math.bas.bg}

\author[D.~F.~M.~Torres]{Delfim F. M. Torres}
\address[D.~F.~M.~Torres]{Center for Research and Development in Mathematics and Applications (CIDMA), 
Department of Mathematics, University of Aveiro, Aveiro 3810--193, Portugal}
\email{delfim@ua.pt}

\AtBeginDocument{{\noindent\small
This is a preprint version of the paper published open access in \emph{Mathematical Biosciences} 
at \url{https://doi.org/10.1016/j.mbs.2024.109373}.}
\vspace{9mm}}

\begin{document}




\begin{abstract}
The COVID-19 pandemic has presented unprecedented challenges worldwide, 
necessitating effective modelling approaches to understand and control 
its transmission dynamics. In this study, we propose a novel approach 
that integrates asymptomatic and super-spreader individuals in a single 
compartmental model. We highlight the advantages of utilizing incommensurate 
fractional order derivatives in ordinary differential equations, 
including increased flexibility in capturing disease dynamics and 
refined memory effects in the transmission process. We conduct a qualitative 
analysis of our proposed model, which involves determining the basic 
reproduction number and analysing the disease-free equilibrium's stability. 
By fitting the proposed model with real data from Portugal and comparing 
it with existing models, we demonstrate that the incorporation of 
supplementary population classes and fractional derivatives significantly 
improves the model's goodness of fit. Sensitivity analysis further provides 
valuable insights for designing effective strategies 
to mitigate the spread of the virus. 
\end{abstract}

\maketitle


\section{Introduction}

The worldwide effects of the COVID-19 pandemic have presented challenges to global public health and economies. 
Refining the understanding of the dynamics of disease transmission can serve in the development of effective 
strategies for controlling and mitigating the spread of not only COVID-19 but also potential future viral 
outbreaks. Over the past years, various mathematical models have been developed to describe the spread 
of the coronavirus SARS-CoV-2, including traditional compartmental models such as SEIR and SIR models 
\cite{MyID:461,MyID:468}. However, these models often assume a homogeneous mixing of individuals and 
do not fully capture the complex dynamics of the disease, particularly the role of asymptomatic 
and super-spreader individuals~\cite{Illingworth2021}.

To overcome these limitations, recent studies have proposed models that incorporate 
asymptomatic \cite{maira1,MyID:469} and super-spreader \cite{MyID:460,wuhan} individuals, 
separately. However, these models may not fully capture the complex interactions between 
these two types of individuals and their impact on disease transmission dynamics. 
In our study, we propose a novel approach that integrates both asymptomatic and super-spreader 
individuals in a single model, taking advantage of fractional calculus 
to refine the model's performance.

Fractional calculus provides a powerful tool for modelling infectious diseases 
like COVID-19 due to its ability to incorporate memory and long-range dependence 
in transmission dynamics~\cite{analysis_covid,FracSIR2017}. In our study, we utilize 
incommensurate fractional order derivatives in ordinary differential equations (ODEs) 
for modelling COVID-19 transmission dynamics, offering two distinct advantages. 
Firstly, incommensurate fractional order derivatives 
allow for greater flexibility in capturing the heterogeneous nature of disease dynamics, 
accounting for factors such as population demographics, social behaviours, and intervention 
measures that can greatly impact the transmission dynamics \cite{JAHANSHAHI2021}. 
This enables the model to accurately represent the complex and evolving nature of disease 
spread in real-world scenarios. Secondly, incommensurate fractional order derivatives provide 
a more refined description of memory effects in the disease transmission process, capturing the 
long-range dependence and persistence observed in real-world data. This allows for more precise 
modelling of memory effects, enhancing the accuracy of the model in capturing the impact 
of past infections on future disease spread. 

The structure of the article is as follows. In Section~\ref{sec:propersties}, 
we present some basic definitions and statements from fractional calculus. 
In Section~\ref{sec:model}, we introduce our proposed model, which includes 
different classes of individuals such as susceptible, exposed, symptomatic 
and infectious, super-spreaders, asymptomatic, hospitalized, recovered, and fatalities. 
In Section~\ref{sec:analysis}, we perform a qualitative analysis of our proposed model, 
which includes determining the basic reproduction number and analysing the disease-free equilibrium. 
In Section~\ref{sec:results}, we present the numerical results, including the fitting of the model 
with real data from Portugal and a comparison with existing models. We also perform a sensitivity 
analysis to assess the impact of various model parameters on the basic reproduction number, 
which provides valuable insights for policy-makers in designing effective strategies 
to mitigate the spread of emergent infectious diseases. We end with Section~\ref{sec:conc}
of conclusions.


\section{Basic definitions and fundamental properties of fractional calculus}
\label{sec:propersties}

We denote the Euclidean norm on $\mathbb{R}^n$ by $\|\cdot\|$.

\begin{definition}
The Caputo incommensurate fractional derivative 
of orders $(\alpha_i)_{i=1, \cdots, n} \in(0,1)$ 
of a function ${x:[0,+\infty)\rightarrow\mathbb R^n}$ is defined by
\[
{}^{\mathtt C}D^{(\alpha_i)}x(t) 
= \begin{cases}
{}^{\mathtt C}D^{\alpha_1}x_1(t),\\
{}^{\mathtt C}D^{\alpha_2}x_2(t),\\
\quad \vdots  \quad \vdots\\
{}^{\mathtt C}D^{\alpha_n}x_n(t),\\
\end{cases}
\]
with
\[
{}^{\mathtt C}D^{\alpha_i}x_i(t)=\frac{1}{\Gamma(1-\alpha_i)} 
\int_{0}^{t} (t-s)^{-\alpha_i} x'_i(s)ds,
\]
where $\Gamma(1-\alpha_i)=\int_0^{\infty}t^{-\alpha_i}\exp(-t)dt$ 
is the Euler Gamma function. 
\end{definition}

Note that the value of the Caputo fractional derivative of the function $x$ 
at time $t$ involves all values of the derivative $x'(s)$ for $s \in [0,t]$, 
and hence it incorporates the history of $x$. 

We have that ${}^{\texttt C}D^{\alpha_i}x_i(t)$ tends to $x'_i(t)$ as $\alpha_i \to 1$. 
In what follows, we recall a fractional Gronwall inequality that is useful to prove 
existence and uniqueness of the solution to the initial value problem \eqref{syst:exis}. 

\begin{lemma}[Gronwall inequality \cite{gronwall}]
\label{thm:gronwall}
Let $\alpha$ be a positive real number, and let $p(\cdot)$ and $u(\cdot)$ 
be non-negative continuous functions on $[a, b]$, and $q(\cdot)$ 
a non-negative, non-decreasing continuous function on $[a, b)$.
If
\[
u(t)\leq p(t) + q(t)\int^t_a(t-s)^{\alpha-1}u(s)ds,
\]
then 
\[
u(t)\leq p(t) + \int^t_a \left[ \sum^{\infty}_{n=1}
\frac{\left(q(t)\Gamma(\alpha ) \right)^n}{\Gamma(n\alpha)} 
(t-s)^{n\alpha-1}p(s)\right]ds
\]
for all $t\in [a,b)$.
\end{lemma}

Let $f: \mathbb{R}^n \rightarrow \mathbb{R}^n,n > 1$ be a vector field.
Consider the following fractional initial value problem:
\begin{equation}
\label{syst:exis}
\begin{cases}
{}^{\mathtt C}D^{(\alpha_i )}x(t)=f(t, x),\quad \alpha_i\in(0,1]\\[3mm]
x(0)=x_0, \, \, x_0\in  \mathbb{R}^n.
\end{cases}
\end{equation}

\begin{theorem}[Existence and uniqueness of solutions]
\label{exist:result}
Assume that the vector field $f$ satisfies the following conditions:
\begin{itemize}
\item $f(t,x)$ is Lebesgue measurable with respect to $t$ on $[0,+\infty)$;
\item $f(t, x)$ and $\displaystyle{\frac{\partial f(t,x)}{\partial x}}$ 
are continuous with respect to $x\in\mathbb{R}^n$;
\item For two positive constants $\omega$ and $\lambda$, 
\begin{equation}
\label{bound}
\lVert f(t, x) \rVert \leq \omega +\lambda \lVert x\rVert, 
\, \, \forall x\in \mathbb{R}^n,\quad a.e.~ t\in[0,+\infty).
\end{equation}
\end{itemize}
Then the initial value problem \eqref{syst:exis} 
has a unique solution on $[0, +\infty )$.
\end{theorem}

\begin{proof}
The initial value problem \eqref{syst:exis} 
can be written in the following vector form:
\begin{equation}
\label{integr}
x(t)= x_0 + \int_0^t {\rm diag}\left(
\frac{(t-s)^{\alpha_i - 1}}{\Gamma(\alpha_i)} \right)f(s, x(s))ds \;.
\end{equation}
Let $a, b>0$ and define the domain
\[
\mathcal{D}= \left\{ (t, y)\in [0, +\infty) 
\times \mathbb{R}^n :  t\leq a, \, \, \|y-x_0\| \leq b\right\}.
\]
Observe that condition~\eqref{bound} implies that the vector field $f$ 
is bounded on $\mathcal{D}$: $\|f(t,x)\|<m,(t,x)\in\mathcal{D}$.

First of all, we must prove the existence and uniqueness of a solution 
to~\eqref{syst:exis} for $t\in[0,h)$ for some positive $h>0$.

The idea of the proof follows the ideas of the proof of Theorem~2.1 
together with Remark~2.3 of \cite{lin}. We show that 
\begin{equation}
\label{eqn:lebesgueint}
{\rm diag}\left(\frac{(t-s)^{\alpha_i - 1}}{\Gamma(\alpha_i)} \right)
f(s, \varphi(s))
\end{equation}
is Lebesgue integrable with respect to $s \in [0, t] ~(t \le h \le a)$, 
provided $\varphi(s)$ is Lebesgue measurable on the interval $[0, h]$.
In fact, since the Euler Gamma function is monotone decreasing on $(0,1]$ 
and $(t-s)^{\alpha_i - 1}<(t-s)^{\alpha_j-1}$ for $\alpha_j<\alpha_i$, we have
\begin{multline}
\label{estim}
\int_0^t\left\lVert {\rm diag}\left(\frac{(t-s)^{\alpha_i - 1}}{\Gamma(\alpha_i)} \right)
f(s, \varphi(s)) ds\right\rVert   \leq \frac{(t-s)^{\alphamin -1}}{\Gamma(\alphamax)}\lVert 
f(s, \varphi(s))\|\le \frac{(t-s)^{\alphamin -1}}{\Gamma(\alphamax)}m,
\end{multline}
where $\alphamin=\min_{i=1,\cdots, n}\alpha_i,\alphamax=\max_{i=1,\cdots, n}\alpha_i$.
H\"older's inequality implies that \eqref{eqn:lebesgueint} is Lebesgue integrable.
Using similar estimates, the remainder of the proof follows Steps~2 and 3 from 
the proof of Theorem~2.1 in \cite{lin} to construct a sequence of vector-valued 
functions $\{\varphi_n(t)\}_{n\in\mathbb{N}}$:
\[
\varphi_n(t)=
\begin{cases}
x_0,\quad 0 \le t \le \displaystyle\frac{h}{n}\; ;\\
x_0 + \displaystyle\int_0^{t-\frac{h}n} {\rm diag}\left(
\frac{(t-s)^{\alpha_i - 1}}{\Gamma(\alpha_i)} \right)
f(s, x(s))ds, \quad \frac{h}{n}\le t\le h\;.
\end{cases}
\] 
Note that in our case we must choose $h<\min\{a,\frac{b\Gamma(\alphamax)\alphamin}{m}\}$ 
to construct the sequence of functions $\{\varphi_n(t)\}$.

Following the reasoning in~\cite{lin}, $\varphi_n(t)$ are continuous in $t\in[0,h]$, 
uniformly bounded, i.e., $(t, \phi_n (t)) \in\mathcal{D}$), and equicontinuous. 
These functions can be shown to converge uniformly to a function $\varphi(t)$ 
on $[0,h]$, which solves~\eqref{syst:exis} on this interval.

Now, it remains to prove that the solution to~\eqref{syst:exis} exists globally. 
By reduction to absurdity, assume that the solution $x$ admits a maximal existence interval, 
denoted by $[0, \Pi)$, $\Pi < + \infty$. By substitution of assumption \eqref{bound} 
into \eqref{integr}, and applying the relation \eqref{estim}, we get the estimation
\[
\lVert x(t)\rVert \leq \lVert x_0 \rVert 
+ \frac{\omega}{\alphamin \Gamma(\alphamax)}|\Pi|^{\alphamin} 
+ \frac{\lambda}{\Gamma(\alphamax)}
\int^t_0 (t-s)^{\alphamin -1}\lVert x(s)\rVert ds.
\]
Since $\lambda$ is a constant, it can be considered a non-negative 
and non-decreasing function. Thus, by applying the Gronwall inequality 
of Lemma~\ref{thm:gronwall}, we have that 
\[
\lVert x(t)\rVert \leq K\left( \lVert x_0 \rVert 
+ \frac{\omega}{\alphamin \Gamma(\alphamax)}|\Pi|^{\alphamin} \right),
\]
with 
\[
K= 1+ \int_0^t \sum_{n=1}^{\infty}\frac{\left( 
\frac{\lambda}{\Gamma(\alphamax)}\Gamma(\alphamin)\right)^n}{
\Gamma(n\alphamin)}(t-s)^{n\alphamin -1}ds.
\]
The rest of the proof follows Theorem~3.1 of reference \cite{lin}. 
\end{proof}


\section{The model}
\label{sec:model}

Transmission of infection from individuals without symptoms is now well 
documented through the COVID-19 pandemic, see e.g. \cite{oha} and references 
therein. It is also well known that some individuals are super-spreaders, 
and might transmit the infection to a large number of healthy people \cite{majra}. 
Therefore, it is important to consider these features of transmission by asymptomatic 
and super-spreader individuals in a single mathematical model. This can be done by 
enhancing the contact rate between healthy and unhealthy individuals in order 
to include both routes of infection. Thus, we extend the model in 
\cite{analysis_covid} by proposing the following new 
Caputo incommensurate fractional-order system:
\begin{equation}
\label{model}
\begin{cases}
\displaystyle{{}^{\textsc c}D^{\alpha_S}S(t) 
= -\beta\frac{I}{N}S-l\beta \frac{H}{N}S
-\beta^{'}\frac{P}{N}S - \beta^{''}\frac{A}{N}S},\\[3mm]
\displaystyle{{}^{\textsc c}D^{\alpha_E}E(t)
= \beta\frac{I}{N}S+l\beta \frac{H}{N}S
+ \beta^{'}\frac{P}{N}S + \beta^{''}\frac{A}{N}S -\kappa E}, \\[3mm]
\displaystyle{{}^{\textsc c}D^{\alpha_I}I(t)
= \kappa \rho_1 E - (\gamma_a + \gamma_i)I-\delta_i I}, \\[3mm]
\displaystyle{{}^{\textsc c}D^{\alpha_P}P(t)
= \kappa \rho_2 E- (\gamma_a + \gamma_i)P-\delta_p P}, \\[3mm]
\displaystyle{{}^{\textsc c}D^{\alpha_A}A(t)
= \kappa (1-\rho_1 - \rho_2)E -\delta_a A },\\[3mm]
\displaystyle{{}^{\textsc c}D^{\alpha_H}H(t)
= \gamma_a (I + P) - \gamma_r H - \delta_h H}, \\[3mm]
\displaystyle{{}^{\textsc c}D^{\alpha_R}R(t)
= \gamma_i (I + P)+ \gamma_r H},\\[3mm]
\displaystyle{{}^{\textsc c}D^{\alpha_F}F(t)
= \delta_i I + \delta_p P + \delta_a A +\delta_h H},
\end{cases}
\end{equation}
where the fractional orders $\alpha_S$, $\alpha_E$, $\alpha_I$, $\alpha_P$, 
$\alpha_A$, $\alpha_H$, $\alpha_H$, $\alpha_F \in (0, 1)$.  
The model subdivides the human individuals into $8$ mutually exclusive classes, 
namely the susceptible class ($S$), exposed class ($E$), symptomatic and infectious class ($I$), 
super-spreader class ($P$), infected but asymptomatic class ($A$), 
hospitalized class ($H$), recovered class ($R$), and fatality class ($F$). Here $N$ 
represents the total population, being given by $N= S+E+I+P+A+H+R+F$. 

Susceptible individuals get infected by the virus through a force of infection 
given by the expression $\displaystyle{\beta\frac{I}{N}S+l\beta \frac{H}{N}S
+ \beta^{'}\frac{P}{N}S + \beta^{''}\frac{A}{N}S}$. This force of infection 
resembles classical ones, where fewer compartments are used \cite{dietz}. 

Notice that $P$ represents super-spreaders among symptomatic individuals only. 
Super-spreaders among asymptomatic individuals are not differentiated, 
and their effects are averaged with non-super-spreaders within that group.

Finally, Table~\ref{tab:parCovid} presents a comprehensive explanation 
of the parameters incorporated in our incommensurate Caputo fractional 
order system \eqref{model}.

\begin{table}[ht!]
{\small{\caption{Description of the parameters used in model~\eqref{model} and their values.}
\label{tab:parCovid}
\begin{tabular}{llll}
\hline
Parameters & Description & Value & Units\\
\hline
\(\beta\) & Infection rate from infected individuals & fitted & day\(^{-1}\)\\
\(\beta'\) & Infection rate due to super-spreaders & fitted & day\(^{-1}\)\\
\(\beta''\) & Infection rate due to asymptomatic individuals & fitted & day\(^{-1}\)\\
\(l\) & Relative transmissibility from hospitalized patients & 1.6054 & dimensionless\\
\(\kappa\) & Rate of exposed individuals becoming infectious & 0.0366 & day\(^{-1}\)\\
\(\rho_1\) & Rate of exposed individuals becoming infected with symptoms & fitted & dimensionless\\ 
\(\rho_2\) & Rate of exposed individuals becoming super-spreaders & fitted & dimensionless\\ 
\(\gamma_a\) & Rate of hospitalization & 0.1375 & day\(^{-1}\)\\
\(\gamma_i\) & Recovery rate of non-hospitalized cases & 0.0560 & day\(^{-1}\)\\ 
\(\gamma_r\) & Recovery rate of hospitalized cases & 0.9180 & day\(^{-1}\)\\ 
\(\delta_i\) & Death rate of infected individuals & 0.0442 & day\(^{-1}\)\\ 
\(\delta_p\) & Death rate of super-spreaders & fitted & day\(^{-1}\)\\
\(\delta_h\) & Death rate of hospitalized individuals & 0.0038 & day\(^{-1}\)\\
\(\delta_a\) & Death rate of asymptomatic individuals & fitted & day\(^{-1}\)\\ \hline
\end{tabular}}}
\end{table}


\section{Qualitative analysis}
\label{sec:analysis}

For biological reasons, let us consider the feasible region   
\[
\Omega = \{(S, E, I, P, A, H, R, F) 
\in \mathbb{R}^8_{+}: S+E+I+P+A+H+R+F\leqslant N \}, 
\]
together with the initial conditions 
\begin{gather}
\label{init_cond}
\begin{gathered}
S(0)\geq 0, \,\, E(0)\geq 0, \, \, I(0)\geq 0, \, \, P(0)\geq 0, \\ 
A(0)\geq 0, \, \,  H(0)\geq 0, \, \, R(0)\geq 0, \, \, F(0)\geq 0.
\end{gathered}
\end{gather}
The following result holds.

\begin{theorem}
\label{theo:uniq}
There exists a unique solution for the initial value problem \eqref{model}--\eqref{init_cond}. 
Moreover, the solution remains in $\Omega $ for all time $t\geq 0$. 
\end{theorem}

\begin{proof}
Let us start by writing the components of the model system \eqref{model} as below: 
\begin{equation}
X(t)=
\begin{pmatrix}
S(t)\\
E(t)\\
I(t)\\
P(t)\\
A(t)\\
H(t)\\
R(t)\\
F(t)
\end{pmatrix}, 
\qquad
F(X)=\begin{pmatrix}
-\beta\frac{I}{N}S-l\beta \frac{H}{N}S-\beta^{'}\frac{P}{N}S
-\beta^{''}\frac{A}{N}S\\
\beta\frac{I}{N}S+l\beta \frac{H}{N}S+ \beta^{'}\frac{P}{N}S 
+ \beta^{''}\frac{A}{N}S -\kappa E\\
\kappa \rho_1 E - (\gamma_a + \gamma_i)I-\delta_i I\\
\kappa \rho_2 E- (\gamma_a + \gamma_i)P-\delta_p P\\
\kappa (1-\rho_1 - \rho_2)E-\delta_a A\\
\gamma_a (I + P) - \gamma_r H - \delta_h H\\
\gamma_i (I + P)+ \gamma_r H\\
\delta_i I + \delta_p P + \delta_a A +\delta_h H
\end{pmatrix}.
\end{equation}
It is not hard to check that the vector field $F$ satisfies the first condition 
of Theorem~\ref{exist:result}. It remains only to show the second condition. 
For this purpose, set
\begin{equation*}
A_1= \begin{pmatrix}
0&0&-\beta & -\beta^{'}&-\beta^{''}&-l\beta &0&0\\
0&0&\beta & \beta^{'}&\beta^{''}&l\beta &0&0\\
0&0&0&0&0&0&0&0\\
0&0&0&0&0&0&0&0\\
0&0&0&0&0&0&0&0\\
0&0&0&0&0&0&0&0\\
0&0&0&0&0&0&0&0\\
0&0&0&0&0&0&0&0\\
\end{pmatrix}, 
\, \, A_2=
\begin{pmatrix}
0&0&0&0&0&0&0&0\\
0&-\kappa&0&0&0&0&0&0\\
0&\kappa \rho_1&-\varpi_i&0&0&0&0&0\\
0&\kappa \rho_2&0&-\varpi_p&0&0&0&0\\
0&\varpi_e&0&0&-\delta_{a}&0&0&0\\
0&0&\gamma_a&\gamma_a&0&-\varpi_h&0&0\\
0&0&\gamma_i&\gamma_i&0&\gamma_r&0&0\\
0&0&\delta_i&\delta_p&\delta_a&\delta_h &0&0\\
\end{pmatrix},
\end{equation*}
with $\varpi_e = \kappa (1-\rho_1-\rho_2)$; 
$\varpi_i = \gamma_a + \gamma_i + \delta_i;\, 
\varpi_p= \gamma_a + \gamma_i + \delta_p$ 
and $\varpi_h= \gamma_r + \delta_h$. Therefore,  
$F$ can be expanded as $F(X)= \frac{S}{N}A_1X + A_2X$, and this leads to 
\[
\lVert F(X) \rVert \leq \lVert A_1X \rVert 
+ \lVert A_2X \rVert < \epsilon +  (A_1 + A_2)\lVert X \rVert,
\]
for any positive constant $\epsilon$. Hence, it follows from Theorem~\ref{exist:result} 
that the model system \eqref{model} subject to \eqref{init_cond} has a unique solution.
\end{proof}

It is of great importance to determine the basic reproduction number $\mathcal{R}_0$ 
for the epidemiological model~\eqref{model}. This quantity represents the number of 
cases one infected case generates on average throughout the infectious period, 
in an otherwise fully susceptible population. Following \cite{MyID:460,wuhan}, 
we prove the following theorem.

\begin{theorem}
The basic reproduction number associated to the model system \eqref{model} is 
\begin{equation}
\label{ro}
\mathcal{R}_0= \frac{\beta \rho_1 (\gamma_a l + \varpi_h)}{\varpi_h \varpi_i} 
+ \frac{\rho_2(\beta \gamma_a l + \beta^{'}\varpi_h)}{\varpi_h \varpi_p} 
+ \frac{\beta^{''}(1-\rho_1 -\rho_2)}{\delta_a}.
\end{equation}
\end{theorem}

\begin{proof}
The proof is done by the well-known next generation matrix approach \cite{r0} 
applied to our model. For this, the following next generation matrices hold:
\begin{equation*}
F= 
\begin{pmatrix}
0&\beta & \beta^{'}&\beta^{''}&l\beta \\
0&0&0&0&0\\
0&0&0&0&0\\
0&0&0&0&0\\
0&0&0&0&0\\
\end{pmatrix}, \qquad 
V=
\begin{pmatrix}
\kappa&0&0&0&0\\
-\kappa \rho_1&\varpi_i&0&0&0\\
-\kappa \rho_2&0&\varpi_p&0&0\\
-\varpi_e&0&0&\delta_{a}&0\\
0&-\gamma_a&-\gamma_a&0&\varpi_h\\
\end{pmatrix}.
\end{equation*}
The basic reproduction ratio $\mathcal{R}_0$ is then 
computed as the spectral radius of $F\cdot V^{-1}$.
\end{proof}

\begin{remark}
The last term in \eqref{ro}, that is, 
$\displaystyle{\frac{\beta^{''}(1-\rho_1 -\rho_2)}{\delta_a}}$, 
quantifies the number of susceptible that one asymptomatic 
individual infects through its infectious lifetime. 
This is an additional term to the basic reproduction number 
obtained in \cite{MyID:460,wuhan}, and it is responsible for 
a greater value to the basic reproduction number for the parameters 
of the new model considered in this paper. 
\end{remark}

Next, we state and prove a theorem related to the stability of the disease 
free equilibrium point 
$$ 
DFE = (N,0,0,0,0,0,0,0).
$$

\begin{theorem}
The disease free equilibrium DFE of system \eqref{model} 
is globally asymptotically stable whenever $\mathcal{R}_0 < 1$. 
\end{theorem}

\begin{proof}
To show global stability of $DFE$, we propose the following Lyapunov function:
\[
V(t)= b_0E(t)+b_1I(t)+b_2P(t)+b_3A(t) + b_4H(t),
\]
where $b_0, \, b_1, \, b_2, \, b_3$ and $b_4$ are positive constants to be determined.

By linearity of the fractional operator ${}^{\textsc c}D^{\alpha}$, we have,
for $\alpha_E=\alpha_I=\alpha_P=\alpha_A=\alpha_H=\alpha$, that 
\[
{}^{\textsc c}D^{\alpha} V(t)= b_0 {}^{\textsc c}D^{\alpha}E(t) 
+ b_1 {}^{\textsc c}D^{\alpha}I(t) + b_2 {}^{\textsc c}D^{\alpha}P(t) 
+ b_3 {}^{\textsc c}D^{\alpha}A(t) + b_4{}^{\textsc c}D^{\alpha}H(t),
\]
and from \eqref{model} it follows that
\begin{align*}
{}^{\textsc c}D^{\alpha}V(t)
&= b_0\left(\beta\frac{I}{N}S+l\beta \frac{H}{N}S+ \beta^{'}\frac{P}{N}S 
+ \beta^{''}\frac{A}{N}S -\kappa E \right) + b_1\left(\kappa \rho_1 E 
- (\gamma_a + \gamma_i)I-\delta_i I \right) \\
&+b_2\left(\kappa \rho_2 E- (\gamma_a + \gamma_i)P-\delta_p P \right) 
+ b_3\left( \kappa (1-\rho_1 - \rho_2)E -\delta_a A \right)\\ 
&+ b_4\left( \gamma_a (I + P) - \gamma_r H - \delta_h H\right).
\end{align*}
Note that, because $0 \leq S\leqslant N$, we have 
\begin{align*}
{}^{\textsc c}D^{\alpha}V(t)
&\leqslant b_0\left(\beta I+l\beta H+ \beta^{'}P + \beta^{''}A 
-\kappa E \right) + b_1\left(\kappa \rho_1 E - \varpi_i I \right)  \\
&+b_2\left(\kappa \rho_2 E- \varpi_p P \right) + b_3\left( \kappa (1-\rho_1 - \rho_2)
E -\delta_a A \right) + b _4\left( \gamma_a (I + P) - \varpi_h H\right).
\end{align*}
As a consequence, we obtain the following reduced inequality relationship: 
\begin{align*}
{}^{\textsc c}D^{\alpha}V(t)
&\leqslant  (b_0\beta + b_4\gamma_a -b_1\varpi_i)I 
+ (b_0\beta l -a_4\varpi_h)H + (b_0\beta^{'} + b_4\gamma -b_2\varpi_p)P  \\
& + (b_0 \beta{''} - b_3\delta_a )A + \kappa(b_1\rho_1 + b_2\rho_2 
+ b_3(1-\rho_1 -\rho_2)-b_0)E.
\end{align*}
Thus, we fix the coefficients $b_0, \, b_1, \, b_2, \, b_3, \, b_4$ as follows:
\begin{gather*}
b_0= \delta_a \varpi_i \varpi_p \varpi_h; 
\quad b_1= \left(\beta + \frac{\beta \gamma_a l}{\varpi_h} \right) \delta_a \varpi_h \varpi_p; 
\quad b_2= \left(\beta^{'} + \frac{\beta \gamma_a l}{\varpi_h} \right) \delta_a \varpi_i \varpi_h;\\ 
b_3= \beta_{''}\varpi_i \varpi_h \varpi_p; \quad b_4= \beta l \delta_a \varpi_i \varpi_p.
\end{gather*}
It is easy to check that our function $V$ is continuous and positive definite for all 
$E(t)>0$, $I(t)>0$, $P(t)>0$, $A(t)>0$ and $H(t)>0$. Moreover, we might also notice that
\begin{gather*}
b_0\beta + b_4\gamma_a -b_1\varpi_i= 0,\quad b_0\beta l 
-a_4\varpi_h = 0, \quad b_0\beta^{'} + b_4\gamma -b_2\varpi_p= 0, 
\quad b_0 \beta{''} - b_3\delta_a =0,
\end{gather*}
and
\begin{equation*}
\begin{split}
b_1&\rho_1 + b_2\rho_2 + b_3(1-\rho_1 -\rho_2)-b_0 = \beta \delta_a \rho_1
\varpi_h\varpi_p + \beta \delta_a \rho_1\gamma_a l\varpi_p + \beta^{'} 
\delta_a \rho_2\varpi_h\varpi_i + \beta \delta_a \rho_2\gamma_a l\varpi_i\\
&\quad + \beta^{''}(1-\rho_1 -\rho_2)\varpi_i \varpi_h \varpi_p -\delta_a \varpi_i \varpi_p \varpi_h\\
&= \left(\frac{\beta \delta_a\rho_1 \varpi_h\varpi_p + \beta \delta_a \rho_1\gamma_al\varpi_p 
+ \beta^{'} \delta_a\rho_2\varpi_h\varpi_i + \beta \delta_a \rho_2\gamma_a l\varpi_i 
+\beta^{''}(1-\rho_1 -\rho_2)\varpi_i \varpi_h \varpi_p }{\delta_a \varpi_i \varpi_p \varpi_h} -1 \right)\\
&\quad \times  \delta_a \varpi_i \varpi_p \varpi_h\\
&=\delta_a \varpi_i \varpi_p \varpi_h \left(R_0 -1 \right) .
\end{split}
\end{equation*}
Altogether, we obtain
\[
{}^{\textsc c}D^{\alpha} V(t)
\leqslant \kappa \delta_a\varpi_i \varpi_p \varpi_h (R_0 -1)E.
\]
Therefore, ${}^{\textsc c}D^{\alpha} V(t)\leqslant 0$ whenever  $R_0 <1$. 
In addition, ${}^{\textsc c}D^{\alpha}V(t)=0$ if, and only if, $E=I=P=H=0$. 
Substituting $(E, I, P, H)=(0,0,0,0)$ into system \eqref{model} leads to 
$$
S(t)=S(0), \quad A(t)=A(0), \quad R(t)=R(0), \quad F(t)=F(0).
$$
Thus, we deduce that $A(0)=R(0)=F(0)=0$ and $S(0)=N$, and the largest compact 
invariant set $\Gamma =\lbrace (S, E, I, P, A, H, R, F)
\in \mathbb{R}^8_{+}:{}^{\textsc c}D^{\alpha} V(t)=0\rbrace$ becomes 
the disease free equilibrium point $DFE$. Finally, by the Lasalle 
invariance principle~\cite{LASALLE196857}, we conclude that the disease 
free equilibrium $DFE$ is globally asymptotically stable.
\end{proof}

Global stability of the incommensurate fractional form of model~\eqref{model} 
is confirmed through Ulam--Hyers stability~\cite{jung2011hyers}, with recent 
COVID-19 model applications in~\cite{Baba2020}. We first discuss the conditions 
necessary to ensure the positiveness of the solutions.

\begin{lemma}
\label{lemm:lips}
The function $\textbf{G}(t,\textbf{X}(t))$ 
fulfills the Lipschitz conditions, specifically:
\begin{equation}
\|\textbf{G}(t, \textbf{X}(t)) - \textbf{G}(t, \textbf{X}^*(t))\| 
\leq \Sigma \|\textbf{X}(t)-\textbf{X}^*(t)\|,
\end{equation}
where
\begin{align*}
\textbf{G}(t, \textbf{X}(t)) 
=& (G_1(t, S(t)), G_2(t, E(t)), G_3(t, I(t)), G_4(t, P(t)),\\
&G_5(t, A(t)), G_6(t, H(t)), G_7(t, R(t)), G_8(t, F(t)))^{\intercal}
\end{align*}
is the vector of functions describing the system's dynamics~\eqref{model}, and
\begin{equation}
\begin{split}
\Sigma=max\{|\beta+l \beta + \beta' + \beta''|, |\kappa|, |\gamma_a 
+ \gamma_i +\delta_i|, |\gamma_a + \gamma_i 
+\delta_p|, |\delta_a|,|\gamma_r + \delta_h|\}.
\end{split}
\end{equation}
\end{lemma}

\begin{proof}
Summarizing that $S(t)$ and $S^*(t)$ 
are couple functions yields the following equality:
\begin{equation}
\|G_1(t,S(t)) - G_1(t, S^*(t))\|
=\Big\Vert\ \left( -\beta\frac{I}{N}-l\beta 
\frac{H}{N}-\beta^{'}\frac{P}{N}-\beta^{''}
\frac{A}{N}\right)(S(t)-S^*(t))\Big\Vert\ .
\end{equation}
By defining
\begin{equation}
\Sigma_1=|\beta+l \beta + \beta' + \beta''|,
\end{equation}
we can infer the inequality
\begin{equation}
\label{eq18}
\|G_1(t,S(t)) - G_1(t, S^*(t))\|\leq\Sigma_1\|S(t)-S^*(t)\|.
\end{equation}
Proceeding similarly for the other functions yields
\begin{align}
\label{eq19}
\|G_2(t,E(t)) - G_2(t, E^*(t))\| &\leq\Sigma_2\|E(t)-E^*(t)\| \notag \\
\|G_3(t,I(t)) - G_3(t, {I}^*(t))\| &\leq\Sigma_3\|I(t)-{I}^*(t)\| \notag \\
\|G_4(t,P(t)) - G_4(t, P^*(t))\| &\leq\Sigma_4\|P(t)-P^*(t)\|  \\
\|G_5(t,A(t)) - G_5(t, A^*(t))\| &\leq\Sigma_5\|A(t)-A^*(t)\| \notag \\
\|G_6(t,H(t)) - G_6(t, H^*(t))\| &\leq\Sigma_6\|H(t)-H^*(t)\| \notag \\
\|G_7(t,R(t)) - G_7(t, R^*(t))\| &=0 \notag \\
\|G_8(t,F(t)) - G_8(t, F^*(t))\| &=0 \notag ,
\end{align}
with the $\Sigma$ values specified as
\begin{align}
\label{eq20}
\Sigma_2 &= | \kappa| \notag ,\\
\Sigma_3 &= |\gamma_a+\gamma_i+\delta_i|\notag ,\\
\Sigma_4 &= |\gamma_a+\gamma_i+\delta_p|  ,\\
\Sigma_5 &= |\delta_a| \notag ,\\
\Sigma_6 &= |\gamma_r+\delta_h| \notag .
\end{align}
This analysis, from equations \eqref{eq18} to \eqref{eq20}, 
shows that all eight functions, $F_i$, meet the Lipschitz condition, 
validating their properties for the system \eqref{model}.
\end{proof}

\begin{theorem}
\label{theo: 1}
Given the conditions of Lemma~\ref{lemm:lips}, if the inequality    
\begin{equation}
\Sigma \max_{i} \frac{T^{\alpha_i}}{\Gamma(\alpha_i +1)}<1,
\quad i={S,E,I, P,A,H,R,F}, 
\end{equation}
is satisfied, then the system \eqref{model} 
admits a unique, positive solution.
\end{theorem}

\begin{proof}
Using the integral form of the solution, we can derive 
the state variables expressed in terms of $F_i$ as follows:
\begin{equation}
\label{eq11}
\begin{cases}
S(t)=S(0) + \frac{1}{\Gamma(\alpha_S)}\int_0^t {(t-\tau)^{\alpha_S-1}G_1(\tau,S(\tau))d\tau},\\
E(t)=E(0) + \frac{1}{\Gamma(\alpha_E)}\int_0^t {(t-\tau)^{\alpha_E-1}G_2(\tau,E(\tau))d\tau},\\
I(t)=I(0) + \frac{1}{\Gamma(\alpha_{I})}\int_0^t {(t-\tau)^{\alpha_{I}-1}G_3(\tau,I(\tau))d\tau},\\
P(t)=P(0) + \frac{1}{\Gamma(\alpha_{P})}\int_0^t {(t-\tau)^{\alpha_{P}-1}G_4(\tau,P(\tau))d\tau},\\
A(t)=A(0) + \frac{1}{\Gamma(\alpha_A)}\int_0^t {(t-\tau)^{\alpha_A-1}G_5(\tau,A(\tau))d\tau},\\
H(t)=H(0) + \frac{1}{\Gamma(\alpha_H)}\int_0^t {(t-\tau)^{\alpha_H-1}G_6(\tau,H(\tau))d\tau},\\
R(t)=R(0) + \frac{1}{\Gamma(\alpha_R)}\int_0^t {(t-\tau)^{\alpha_R-1}G_7(\tau,R(\tau))d\tau},\\
F(t)=F(0) + \frac{1}{\Gamma(\alpha_F)}\int_0^t {(t-\tau)^{\alpha_F-1}G_8(\tau,F(\tau))d\tau}.
\end{cases}
\end{equation}
By applying the Picard iteration~\cite{Boudaoui2021} 
to Equation~\eqref{eq11}, we obtain the subsequent equations:
\begin{equation}
\begin{cases}
S_{n+1}(t)=S(0) + \frac{1}{\Gamma(\alpha_S)}\int_0^t {(t-\tau)^{\alpha_S-1}G_1(\tau,S(\tau))d\tau},\\
E_{n+1}(t)=E(0) + \frac{1}{\Gamma(\alpha_E)}\int_0^t {(t-\tau)^{\alpha_E-1}G_2(\tau,E(\tau))d\tau},\\
I_{{n+1}}(t)=I(0) + \frac{1}{\Gamma(\alpha_{I})}\int_0^t {(t-\tau)^{\alpha_{I}-1}G_3(\tau,I(\tau))d\tau},\\
P_{{n+1}}(t)=P(0) + \frac{1}{\Gamma(\alpha_{P})}\int_0^t {(t-\tau)^{\alpha_{P}-1}G_4(\tau,P(\tau))d\tau},\\
A_{n+1}(t)=A(0) + \frac{1}{\Gamma(\alpha_A)}\int_0^t {(t-\tau)^{\alpha_A-1}G_5(\tau,A(\tau))d\tau},\\
H_{n+1}(t)=H(0) + \frac{1}{\Gamma(\alpha_H)}\int_0^t {(t-\tau)^{\alpha_H-1}G_6(\tau,H(\tau))d\tau},\\
R_{n+1}(t)=R(0) + \frac{1}{\Gamma(\alpha_R)}\int_0^t {(t-\tau)^{\alpha_R-1}G_7(\tau,R(\tau))d\tau},\\
F_{n+1}(t)=F(0) + \frac{1}{\Gamma(\alpha_F)}\int_0^t {(t-\tau)^{\alpha_F-1}G_8(\tau,F(\tau))d\tau}.\\
\end{cases}
\end{equation}
Therefore, the system solution \eqref{model} can be expressed in the form
\[
\textbf{X}(t) = \mathcal{P}(\textbf{X}(t)),
\]
where $\textbf{X}(t) = (S(t), E(t), I(t), P(t), A(t), H(t), R(t), F(t))^{\intercal}$ 
is the state vector, and $\mathcal{P}: C([0, T], \mathbb{R}^8) \to C([0, T], \mathbb{R}^8)$ 
denotes the Picard operator. This operator is defined as follows:
\begin{equation}
\begin{split}
\mathcal{P}(\textbf{X}(t)) =\textbf{X}(0) 
&+ \int_0^t \mathrm{diag}\Bigg(  \frac{(t-\tau)^{\alpha_{S}-1}}{\Gamma(\alpha_{S})}, 
\frac{(t-\tau)^{\alpha_{E}-1}}{\Gamma(\alpha_{E})}, 
\frac{(t-\tau)^{\alpha_{I}-1}}{\Gamma(\alpha_{I})}, 
\frac{(t-\tau)^{\alpha_{P}-1}}{\Gamma(\alpha_{P})},\\
& \frac{(t-\tau)^{\alpha_{A}-1}}{\Gamma(\alpha_{A})}, 
\frac{(t-\tau)^{\alpha_{H}-1}}{\Gamma(\alpha_{H})},
\frac{(t-\tau)^{\alpha_{R}-1}}{\Gamma(\alpha_{R})}, 
\frac{(t-\tau)^{\alpha_{F}-1}}{\Gamma(\alpha_{F})}\Bigg) 
\mathbf{G}(\tau, \mathbf{X}(\tau)) \, d\tau.
\end{split}
\label{eq:longEquation}
\end{equation}
Simultaneously, we encounter the series of inequalities below:
\begin{equation}
\begin{split}
\|\mathcal{P}(\textbf{X}(t)) - \mathcal{P}(\textbf{X}^*(t))\|
&=\Big\Vert\ 
\int_0^t \mathrm{diag}\Bigg( \frac{(t-\tau)^{\alpha_{S}-1}}{
\Gamma(\alpha_{S})}, \frac{(t-\tau)^{\alpha_{E}-1}}{\Gamma(\alpha_{E})}, 
\frac{(t-\tau)^{\alpha_{I}-1}}{\Gamma(\alpha_{I})}, \\ 
&\frac{(t-\tau)^{\alpha_{P}-1}}{\Gamma(\alpha_{P})}, 
\frac{(t-\tau)^{\alpha_{A}-1}}{\Gamma(\alpha_{A})}, 
\frac{(t-\tau)^{\alpha_{H}-1}}{\Gamma(\alpha_{H})}, 
\frac{(t-\tau)^{\alpha_{R}-1}}{\Gamma(\alpha_{R})},
\frac{(t-\tau)^{\alpha_{F}-1}}{\Gamma(\alpha_{F})} \Bigg) \\ 
& \times \left( \mathbf{G}(\tau, \mathbf{X}(\tau)) 
- \mathbf{G}(\tau, \mathbf{X}^*(\tau)) \right) \, d\tau \Big\Vert\ \\ 
&\leq \Big\Vert\ 
\int_0^t \mathrm{diag}\Bigg( \frac{(t-\tau)^{\alpha_{S}-1}}{\Gamma(\alpha_{S})}, 
\frac{(t-\tau)^{\alpha_{E}-1}}{\Gamma(\alpha_{E})},
\frac{(t-\tau)^{\alpha_{I}-1}}{\Gamma(\alpha_{I})}, \\ 
&\frac{(t-\tau)^{\alpha_{P}-1}}{\Gamma(\alpha_{P})}, 
\frac{(t-\tau)^{\alpha_{A}-1}}{\Gamma(\alpha_{A})}, 
\frac{(t-\tau)^{\alpha_{H}-1}}{\Gamma(\alpha_{H})}, 
\frac{(t-\tau)^{\alpha_{R}-1}}{\Gamma(\alpha_{R})},
\frac{(t-\tau)^{\alpha_{F}-1}}{\Gamma(\alpha_{F})}\Bigg)
d\tau  \Big\Vert\ \\  
&\times \sup_{\tau \in [0,T]} \| \mathbf{G}(\tau, \mathbf{X}(\tau)) 
- \mathbf{G}(\tau, \mathbf{X}^*(\tau)) \| \\ 
&\leq \max_{i= S,E,I, P,A,H,R,F} \int_0^t 
\frac{(t-\tau)^{\alpha_{i}-1}}{\Gamma(\alpha_{i})} d\tau 
\sup_{\tau \in [0,T]} \| \mathbf{G}(\tau, \mathbf{X}(\tau)) 
- \mathbf{G}(\tau, \mathbf{X}^*(\tau)) \| \\
&\leq \Sigma  \max_{i= S,E,I, P,A,H,R,F} \frac{T^{\alpha_i}}{\Gamma(\alpha_i +1)}  
\sup_{\tau \in [0,T]} \| \mathbf{X}(\tau) - \mathbf{X}^*(\tau) \|.
\end{split}
\end{equation}
Given that $\Sigma \max_{i \in \{S,E,I, P,A,H,R,F\}} 
\frac{T^{\alpha_i}}{\Gamma(\alpha_i +1)} < 1$ for $t \leq T$, 
the operator $\mathcal{P}$ is established as a contraction. Consequently, 
system \eqref{model} is guaranteed to have a unique and positive 
solution, thereby completing the proof.
\end{proof}

To provide a foundation for the subsequent discussion, 
we introduce the inequality expressed as follows:
\begin{equation}
\label{eq:10}
|{}^{C}D^{\alpha}_{0+}\textbf{X}(t)-\textbf{G}(t, 
\textbf{X}(t))|\leq \epsilon, \quad t\in [0,T],
\end{equation}
in which $\alpha=\{\alpha_S,\alpha_E,\alpha_I,\alpha_P,\alpha_A,
\alpha_H,\alpha_R,\alpha_F\}$ and $\epsilon=\{\epsilon_i | i=1,\ldots,8\}$. 
We say a function $\Bar{\textbf{X}}\in \mathbb{R}^8_+$ is a solution of 
\eqref{eq:10} if, and only if, there exists a perturbation 
$h \in \mathbb{R}^8_+$ satisfying
\begin{itemize}
\item[1.] $|h(t)| \leq \epsilon$
\item[2.] ${}^{C}D^{\alpha}_{0+}\Bar{\textbf{X}}(t)
=\textbf{G}(t, \Bar{\textbf{X}}(t))+h(t), \quad t\in [0,T]$.
\end{itemize}
Notably, by applying Equation \eqref{eq11} alongside property 2 mentioned above, 
straightforward simplification reveals that any function $\bar{\mathbf{X}} 
\in \mathbb{R}^8_+$ meeting the conditions of Equation \eqref{eq:10} 
likewise fulfills the following associated integral inequality:
\begin{equation}
\label{eq:32}
|\Bar{\textbf{X}}(t)-\Bar{\textbf{X}}(0)
- \frac{1}{\Gamma(\alpha)}\int_0^t(t-\tau)^{\alpha-1}
\textbf{G}(\tau,\Bar{\textbf{X}}(\tau))| 
\leq \frac{T^\alpha}{\Gamma(\alpha+1)} \epsilon.
\end{equation}

Let $\mathfrak{G}=C([0,T];\mathbb{R})$ denote the Banach space 
of all continuous functions from $[0,T]$ to $\mathbb{R}$ equipped 
with the norm $\|\textbf{X}\|_{\mathfrak{G}}=\sup_{t\in [0,T]}\{|\textbf{X}|\}$, 
where $|\textbf{X}|=|S(t)|+|E(t)|+|I_A(t)|+|I_S(t)|+|R(t)|+|D(t)|+|W(t)|$.

The fractional order model \eqref{model} achieves Ulam--Hyers stability 
if there are some $\Sigma > 0$ ensuring that, for any given 
$\bar{\epsilon}=\mathbb{R}_+$, and for every solution $\Bar{\textbf{X}}$ 
meeting the conditions of \eqref{eq:10}, a corresponding solution 
$\textbf{X}$ to \eqref{model} can be found where 
\begin{equation}
\|\Bar{\textbf{X}}(t) - \textbf{X}(t) \|_{\mathfrak{G}} 
\leq \Sigma \bar{\epsilon}, \quad t \in [0,T].
\end{equation}
Moreover, this model is deemed to be generalized Ulam--Hyers stable 
if a continuous function $\Sigma_{G}: \mathbb{R}_+ \to \mathbb{R}_+$ exists, 
satisfying $\Sigma_G(0)=0$. This condition requires that, for any solution 
$\Bar{\textbf{X}}$ of \eqref{eq:10}, there must be a corresponding 
solution $\textbf{X}$ of \eqref{model} for which
\begin{equation}
\|\Bar{\textbf{X}}(t) - \textbf{X}(t) \|_{\mathfrak{G}} 
\leq \Sigma_G \bar{\epsilon}, \quad t \in [0,T].
\end{equation}
We proceed to detail the stability results for the fractional order model.

\begin{theorem}
Assuming the conditions and conclusions of Lemma~\ref{lemm:lips} 
and Theorem~\ref{theo: 1} are satisfied, i.e. $\Sigma \max_i
\frac{T^{\alpha_i}}{\Gamma({\alpha_i}+1)}<1$, it follows that 
the model specified in \eqref{model} exhibits generalized Ulam--Hyers stability.
\end{theorem}

\begin{proof}
Given that $\textbf{X}$ is a unique solution to \eqref{model} confirmed 
by Lemma~\ref{lemm:lips} and Theorem~\ref{theo: 1}, and $\Bar{\textbf{X}}$ 
meets the criteria of \eqref{eq:10}, reference to equations \eqref{eq11} 
and \eqref{eq:32} leads us to conclude that for any $\epsilon \in \mathbb{R}_+^8$ 
and $t \in [0, T]$, the following relationship holds:
\begin{equation}
\begin{split}
\| \Bar{\textbf{X}} - \textbf{X} \|_{\mathfrak{G}}  
&= \sup_{t \in [0,T]} |\Bar{\textbf{X}} - \textbf{X}| \\
&= \sup_{t \in [0,T]} \left| \Bar{\textbf{X}} - {\textbf{X}}_0 
- \frac{1}{\Gamma(\alpha)} \int_0^t (t - \tau)^{\alpha-1} 
\textbf{G}(t, \textbf{X}(\tau)) d\tau \right| \\
& \leq \sup_{t \in [0,T]} \left| \Bar{\textbf{X}}(t) 
- \Bar{\textbf{X}}_0 - \frac{1}{\Gamma(\alpha)} 
\int_0^t (t - \tau)^{\alpha-1} \textbf{G}(t, \Bar{\textbf{X}}(\tau)) d\tau \right| \\
& \quad + \sup_{t \in [0,T]}  \frac{1}{\Gamma(\alpha)} 
\int_0^t (t - \tau)^{\alpha-1} \left| \textbf{G}(t, \Bar{\textbf{X}}(\tau)) 
- \textbf{G}(t, \textbf{X}(\tau)) \right| d\tau  \\
& \leq \frac{\epsilon T^\alpha}{\Gamma(\alpha+1)} 
+ \frac{\Sigma}{\Gamma(\alpha)} \sup_{t \in [0,T]} 
\int_0^t (t - \tau)^{\alpha-1} |\Bar{\textbf{X}}(\tau) 
- {\textbf{X}}(\tau)| d\tau \\
& \leq \max_{i= S,E,I, P,A,H,R,F} \left( \frac{\bar{\epsilon} 
T^{\alpha_i}}{\Gamma(\alpha_i+1)} + \frac{\Sigma T^{\alpha_i}}{
\Gamma(\alpha_i+1)} \|\Bar{\textbf{X}}(\tau) 
- {\textbf{X}}(\tau)\|_{\mathfrak{G}} \right),
\end{split}
\end{equation}
where, $\displaystyle{\bar{\epsilon}=\max_i \epsilon_i}$. 
From this, we derive that $\displaystyle{\|\Bar{\textbf{X}} 
- \textbf{X}\|_{\mathfrak{G}} \leq \Sigma_G \bar{\epsilon}}$, 
with $\Sigma_G$ defined as 
\[
\Sigma_G=\displaystyle{\max_{i= S,E,I, P,A,H,R,F}
\frac{T^{\alpha_i}}{\Gamma(\alpha_i+1)-T^{\alpha_i} \Sigma}}.
\]
The proof is complete.
\end{proof}


\section{Numerical results}
\label{sec:results}

This section presents the efficiency of our proposed model in simulating 
the dynamics of COVID-19 transmission, which integrates the variables of 
asymptomatic and super-spreader individuals. The model is calibrated using 
real-world data of Portugal obtained from the Center for Systems Science 
and Engineering (CSSE) at Johns Hopkins University \cite{DataCSSE}. 
The initial conditions for the model are described as follows:
\begin{align}
\begin{split}
&
N_0=10280000/1363,\\
& 
S_0=N_0 - 5,\;E_0=0,\;I_0=4,\;P_0=1,\;A_0=0,\;H_0=0,\;R_0=0,\;F_0=0,
\end{split}
\end{align}
while Table~\ref{tab:parCovid} provides the values 
of the model parameters used in the analysis. 

We evaluate the performance of the proposed model \eqref{model} 
in the context of the absence of super-spreaders or asymptomatic individuals. 
Our model incorporates additional coefficients to account for the contribution 
of asymptomatic and super-spreader individuals to the transmission of the disease. 
Thus, we consider six variations of our proposed model \eqref{model}, each with 
different constraints and orders of derivatives, to compare their performance 
in modelling. The first three models are named M1, M2, and M3, and they 
utilize integer-order derivatives.

Model M1 lacks the influence of super-spreader individuals $P$ and their 
related coefficients, $\beta'$, $\rho_2$, and $\delta_p$. Model M2 excludes 
asymptomatic individuals and their coefficients, meaning compartment $A$ 
and its related coefficients $\beta''$ and $\delta_a$ are set to zero. 
In this case, instead of using $\rho_2$, we simply use $1-\rho_1$, 
thus eliminating the need for the parameter $\rho_2$. Model M3 
represents our proposed model with integer-order derivatives.

\begin{table}[ht!]
\caption{Optimized values of parameters for models M1, M2, and M3, 
with integer orders and the order derivatives, denoted as 
\(\alpha_{S, E, I, P, A, H, R, F}\) for models FM1, FM2, and FM3, 
obtained by fitting COVID-19 data from Portugal and evaluated based 
on their root mean square deviation (RMSD) errors.}
\begin{tabular}{ccccccc} \hline
Parameters & M1 & FM1 & M2 & FM2 & M3 & FM3\\ \hline
$\beta$ & 3.3687 & 3.4956 & 3.1300 & 3.5122 &  3.1221 & 2.9763 \\
$\beta'$ & -  &  - & 7.9889 & 7.7403 & 5.0792 & 7.8959\\
$\beta''$ & 4.9091 & 5.0111 & - & - & 4.6831 & 2.7722 \\
$\rho_1$ & 0.7010 &  0.6002 & 0.8760 &  0.8311 & 0.5990 & 0.5273 \\
$\rho_2$ & -  &  - & - & - & 0.1500 & 0.2509 \\
$\delta_p$ & -  &  - & 0.0197 & 0.0400 & 0.0010 & 0.0017 \\
$\delta_a$ & 0.0055  &  0.0061 & - & - & 0.0097 & 0.0204 \\
$\alpha_S$ & -  & 0.9035  & -& 1.0000 & -& 0.9786 \\
$\alpha_E$ & -  & 1.0000  & -& 0.9776 & -& 1.0000\\
$\alpha_I$ & -  &  1.0000 & -& 0.8932 & -& 1.0000\\
$\alpha_P$ & -  &  0.8500 & -& 1.0000 & -& 1.0000\\
$\alpha_A$ & -  &  1.0000 & -& 0.8500 & -& 1.0000\\
$\alpha_H$ & -  &  1.0000 & -& 1.0000 & -& 1.0000\\
$\alpha_F$ & -  &  0.8998 & -& 1.0000 & -& 0.9399 \\ \hline 
RMSD &  123.3005 & 107.4498 & 116.4057  & 109.8153 & 108.2076  & 99.3987 \\ \hline
\end{tabular}
\label{tab:FittedPara}
\end{table}

Furthermore, we compare the fractional-order forms of the models, denoted as 
FM1, FM2, and FM3, respectively. The following box summarizes all the models together.

\begin{tcolorbox}[colback=white!10!white, title=Summary of the Models]
\textbf{Integer-Order Models}
\begin{itemize}
\item[\textbf{M1}]: Model \eqref{model} lacking super-spreader 
individuals ($P=0$, $\beta'$, $\rho_2$, $\delta_p = 0$)
\item[\textbf{M2}]: Model \eqref{model} lacking asymptomatic 
individuals ($A=0$, $\beta''$, $\delta_a = 0$, $\rho_2 = 1-\rho_1$)
\item[\textbf{M3}]: Model \eqref{model} with super-spreader 
and asymptomatic individuals
\end{itemize}
\vspace{1em}
\textbf{Fractional-Order Models}
\begin{itemize}
\item[\textbf{FM1}]: Fractional-order version of M1
\item[\textbf{FM2}]: Fractional-order version of M2
\item[\textbf{FM3}]: Fractional-order version of M3
\end{itemize}
\end{tcolorbox}
We obtain the values of parameters, including $\beta$, $\beta'$, 
$\beta''$, $\rho_1$, $\rho_2$, $\delta_p$, and $\delta_a$, 
by fitting the models to the data. In the case of fractional models 
(FM1, FM2, and FM3), we optimize the values of order derivatives 
and parameters to achieve the best possible fit. The resulting fitted 
values and their root mean square deviation (RMSD) for all models can be 
found in Table~\ref{tab:FittedPara}, such that RMSD
$(x,\widehat{x})=\sqrt{\frac{1}{n} \sum_{t=1}^n(x_t-\widehat{x}_t)^2}$, 
where, $n$ is the number of data points, $x$ approximated values, 
and $\widehat{x}$ real values. 

\begin{figure}[ht!]
\centering
\includegraphics[width=\textwidth]{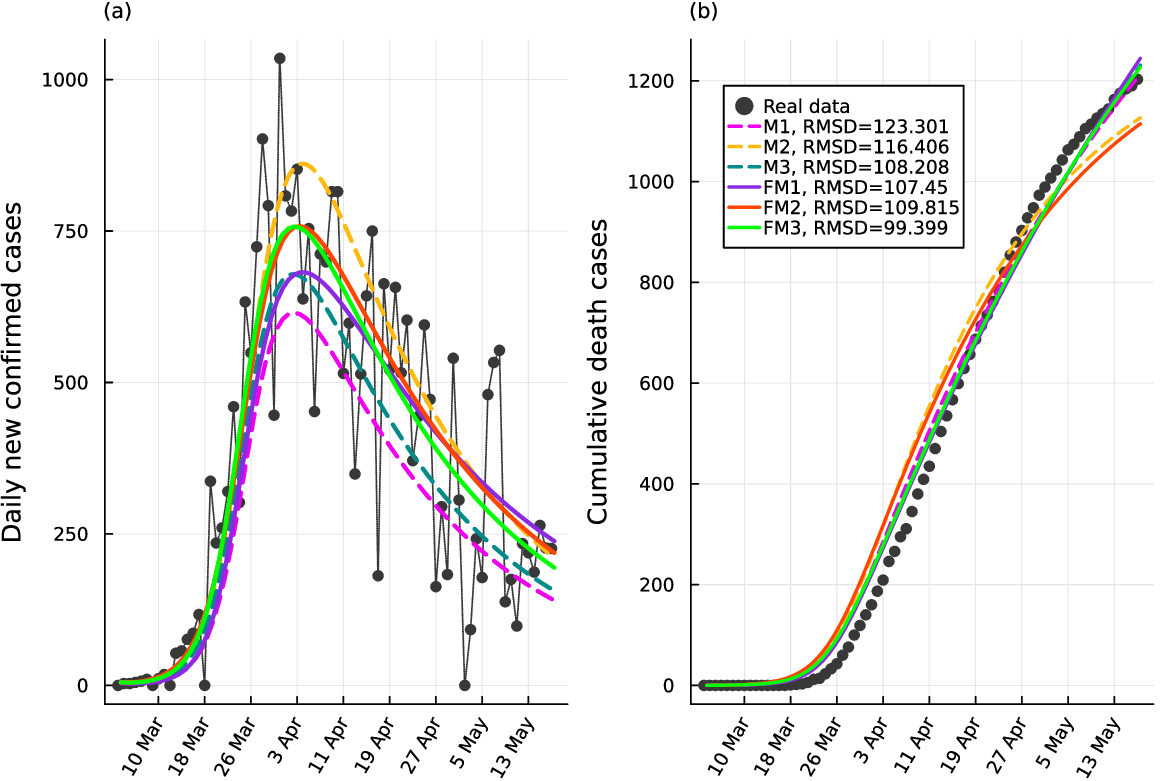}
\caption{Illustration of the accuracy of the studied models in fitting 
(a) daily new confirmed cases and 
(b) cumulative death cases, as retrieved from CSSE~\cite{DataCSSE}. 
The estimation of variables $I+P+H$ and $F$ is shown. Data points are 
represented by black circles. Errors are evaluated using the root mean square deviation (RMSD).
Model M1, which lacks super-spreaders and uses integer orders, is shown to have 
an inferior performance with an error of 123.30. The fractional version of this model, 
FM1 (solid purple line), improves accuracy significantly, reducing the error to 107.45.
Model M2, which lacks asymptomatic individuals (dashed orange line), performs better 
than M1 with an error of 116.41. However, its fractional version, FM2 (solid red line), 
shows a smaller improvement compared to FM1, with an error of 109.81.
Our proposed model \eqref{model} with integer orders, M3 (dashed green line), 
which includes both super-spreader and asymptomatic compartments, performs better 
than both M1 and M2, with an error of 108.21. The fractional version of this model, 
FM3 (solid green line), shows the best performance, reducing the error to 99.40.}
\label{fig:fitting}
\end{figure}

\begin{figure}[ht!]
\centering
\includegraphics[width=\textwidth]{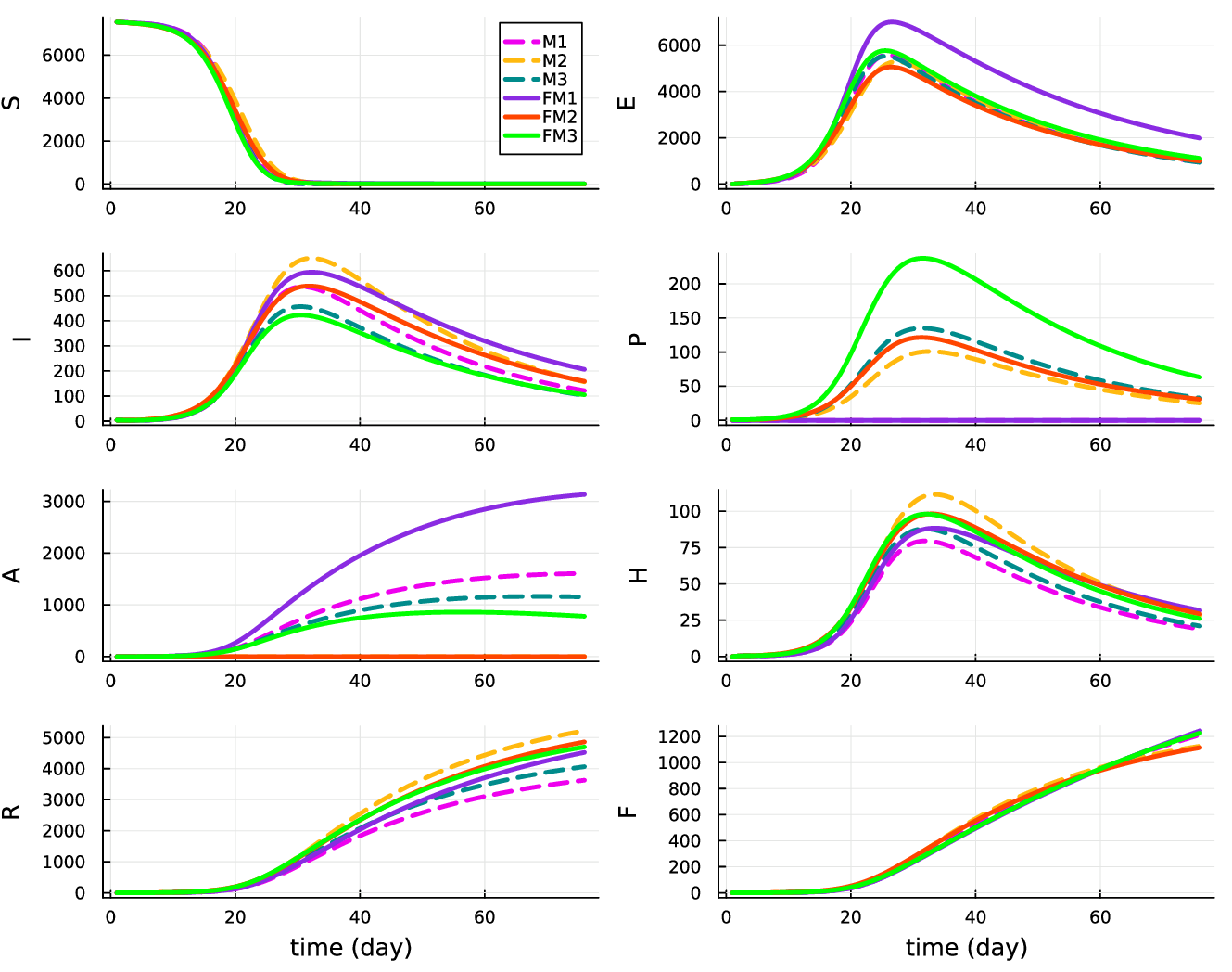}
\caption{Comparison of individual dynamics from simulations of six studied models: 
Model M1 excludes super-spreaders (with parameters $\beta'$, $\rho_2$, $\delta_p = 0$), 
and uses integer-order derivatives. Model M2 excludes asymptomatic individuals (with parameters 
$\beta''$, $\delta_a = 0$, $\rho_2 = 1-\rho_1$), and also uses integer-order derivatives. 
Model M3 is the proposed model \eqref{model} that includes both super-spreader 
and asymptomatic compartments, using integer-order derivatives. The fractional-order 
models are as follows: FM1 is the fractional-order version of M1, FM2 is the 
fractional-order version of M2, and FM3 is the fractional-order version of M3. 
The values for parameters and orders of the derivatives used in simulations 
are provided in Tables~\ref{tab:parCovid} and \ref{tab:FittedPara}.}
\label{fig:allplots}
\end{figure}

Figure~\ref{fig:fitting} illustrates the models' effectiveness in fitting daily 
new confirmed cases and cumulative death cases. The comparison shows that Model 
M1, with an error of 123.30, performs worse than the other models. However, 
the fractional form of this model, FM1, improves accuracy, reducing the error to 107.45.

Model M2 shows better performance with an error of 116.41. The fractional form, FM2, 
further improves the fit, achieving an error of 109.81. This improvement 
is less significant than that seen in FM1.

Our proposed model \eqref{model} with integer-order derivatives, M3, 
which includes both asymptomatic and super-spreader compartments, outperforms 
M1 and M2, with an error of 108.21. The fractional-order version, FM3, 
further enhances accuracy, reducing the error to 99.40, making it 
the best-performing model among all.

Figure~\ref{fig:allplots} presents a comprehensive comparison of the individual 
dynamics obtained from simulations of the six studied models. The simulations 
were conducted using the fitted orders of derivatives and parameters, as outlined 
in Tables~\ref{tab:parCovid} and \ref{tab:FittedPara}. The figure visually 
illustrates the distinct behaviours of the models for different individual cases, 
highlighting the impact of the chosen order derivatives 
and parameters on the dynamics of the disease spread.


\subsection{Sensitivity Analysis of the Basic Reproduction Number}

This section aims to elucidate the impact of varying variables on the propagation 
of infectious illnesses. By evaluating the sensitivity of $\mathcal{R}_0$ to different parameters, 
such as social distancing measures, policy-makers can make informed decisions about public 
health interventions. Additionally, this approach enables the identification of parameters 
with the greatest impact on disease spread. To determine the effects of minor adjustments 
in parameter $\mathcal{P}$ on the basic reproduction number $\mathcal{R}_0$, we introduce 
the forward normalized sensitivity index of $\mathcal{R}_0$ for $\mathcal{P}$, 
which can be expressed as follows:
\begin{equation*}
\mathcal{S}_{\mathcal{P}}^{\mathcal{R}_0}
=\frac{\partial \mathcal{R}_0}{\partial \mathcal{P}}\frac{\mathcal{P}}{\mathcal{R}_0}.
\end{equation*}

We have computed the sensitivity indices for $\mathcal{R}_0$, 
associated with the estimated parameters for models, 
and their corresponding values are presented 
in Table~\ref{tab:Sensitivity}.

\begin{table}[th!]
\caption{Sensitivity index of $\mathcal{R}_0$ across all parameters 
($\mathcal{S}_{\mathcal{P}}^{\mathcal{R}_0}$) and models.}
\begin{tabular}{ccccccc}
\hline
Parameters & \texttt{M1} & \texttt{FM1} & \texttt{M2} 
& \texttt{FM2} & \texttt{M3} & \texttt{FM3} \\ \hline
$\beta$       & 0.0435 & 0.0324 & 0.7601  & 0.7380  & 0.0766 & 0.1843  \\
$\beta'$      &  -       & -       & 0.2399  & 0.2620  & 0.0290 & 0.2055  \\
$\beta''$     & 0.9565  & 0.9676  & -       & -       & 0.8944   & 0.6102   \\
$l$       & 0.0084 & 0.0063 & 0.1651  & 0.1656  & 0.0182 & 0.0506 \\
$\rho_1$     & -2.1883  & -1.4190  & -1.1173  & -0.7190 & -2.0742 & -1.2845    \\
$\rho_2$     & -  &     - &    - &     -     &  -0.5034    & -0.4660 \\
$\gamma_a$   &-0.0161   &-0.0120  &  -0.4127   &-0.3990    & -0.0450 & -0.1937\\
$\gamma_i$   &-0.0098 & -0.0073  &-0.2326    &  -0.2272   &  -0.0255   & -0.0986\\
$\gamma_r$   & -0.0083 & -0.0062 & -0.1627   &  -0.1632  &   -0.0180  &  -0.0499\\
$\delta_i$   &  -0.0078 & -0.0058 & -0.1317     &-0.1267  &   -0.0129 &   -0.0296\\
$\delta_p$     & - &         - &       -0.0231 &   -0.0477  &  -0.0002  &-0.0019\\
$\delta_h$    &-3.14e-5  &-2.34e-5  & -0.0006  &-0.0006 & -6.81e-5   &-0.0002\\
$\delta_a$   & -0.3415 &   -0.3674  &   -       &    -      &    -0.4403  &   -0.4096\\ \hline
\end{tabular}
\label{tab:Sensitivity}
\end{table}

It is clear that if $\mathcal{R}_0$ increases concerning $\mathcal{P}$, 
then the sensitivity index of $\mathcal{S}_{\mathcal{P}}^{\mathcal{R}_0}$ 
is positive; and if $\mathcal{R}_0$ decreases concerning $\mathcal{P}$, 
then the sensitivity index is negative. To elaborate further, 
a sensitivity index of $\mathcal{S}_{\beta''}^{\mathcal{R}_0}=0.9565$ implies 
that an increase of 1\% in $\beta''$, keeping all other parameters constant, 
will lead to a 0.9565\% increase in $\mathcal{R}_0$. Similarly, 
$\mathcal{S}_{\rho_1}^{\mathcal{R}_0}=-2.1883$ means that increasing 
the parameter $\rho_1$ by 1\%, while holding all other parameters constant, 
will cause a decrease in the value of $\mathcal{R}_0$ by 2.1883\%. 

According to Table~\ref{tab:Sensitivity}, the most influential parameter 
for models M1, FM1, M3, and FM3, contributing to the increase in the value 
of $\mathcal{R}_0$, is $\beta''$, whereas it is $\beta$ in models M2 and FM2.
The most influential parameter contributing to its reduction is $\rho_1$ for all models. 

\subsection{Numerical methods and implementation}

We conducted all the numerical analyses using the programming language \texttt{Julia} 
and the high-performance computing system \texttt{PUHTI} at the Finnish IT Center for Science (CSC). 
To solve fractional differential equations, we made use of the \texttt{FdeSolver.jl} 
package (v 1.0.7) that applies predictor-corrector algorithms and product-integration 
rules~\cite{FdeSolver}. Parameter estimation was accomplished through Bayesian inference 
and Hamiltonian Monte Carlo (HMC) using \texttt{Turing.jl}, while ODEs were solved 
with \texttt{DifferentialEquations.jl}. The order of derivatives was optimized 
through the function (L)BFGS, based on the (Limited-memory) 
Broyden--Fletcher--Goldfarb--Shanno algorithm 
from the \texttt{Optim.jl} package and \texttt{FdeSolver.jl}.


\section{Conclusion}
\label{sec:conc}

Our study presents an innovative approach to modelling COVID-19 transmission 
dynamics by integrating asymptomatic and super-spreader individuals into a single 
model using fractional calculus. Furthermore, we have conducted a qualitative 
analysis of our proposed model, which includes determining the basic reproduction 
number and analysing the disease-free equilibrium. Our findings emphasize the 
benefits of incommensurate fractional order derivatives, such as increased 
flexibility in capturing disease dynamics and refined memory effects 
in the transmission process. By fitting the proposed model with real data 
from Portugal and comparing it with existing models, we demonstrate that 
including supplementary coefficients and fractional derivatives enhances 
the model's goodness of fit. Sensitivity analysis further provides valuable 
insights for policy-makers in designing effective strategies to mitigate 
the spread of COVID-19.  Overall, our study contributes to the literature 
on fractional modelling of COVID-19 transmission 
and has potential implications for understanding 
and controlling the spread of infectious diseases.


\section*{Statements and Declarations}

\subsection*{Competing Interests}

The authors declare no conflicts of interest.


\subsection*{Data and Code Availability}

All computational results for this paper are available on GitHub, 
and accessible via the permanent Zenodo DOI: \url{https://doi.org/10.5281/zenodo.14607808}.


\subsection*{Funding}

This study has been supported by the Academy of Finland (330887 to MK, LL) 
and the UTUGS graduate school of the University of Turku (to MK). 
FN is supported by the Bulgarian Ministry of Education and Science, 
Scientific Programme ``Enhancing the Research Capacity in Mathematical Sciences (PIKOM)'', 
Contract No. DO1--67/05.05.2022. 
DFMT is supported by FCT (Funda\c{c}\~{a}o para a Ci\^{e}ncia e a Tecnologia)
through CIDMA projects UIDB/04106/2020 (\url{https://doi.org/10.54499/UIDB/04106/2020})
and UIDP/04106/2020 (\url{https://doi.org/10.54499/UIDP/04106/2020}),
and the CoSysM3 project 2022.03091.PTDC (\url{https://doi.org/10.54499/2022.03091.PTDC}). 


\subsection*{Acknowledgements}

The authors wish to acknowledge CSC-IT Center for Science, Finland, 
for computational resources and high-speed networking.



\end{document}